\documentclass[10pt]{article}
\usepackage{hyperref}
\usepackage{amsmath,amssymb,amsthm}
\usepackage{geometry}
\usepackage{graphicx}
\usepackage{booktabs}
\usepackage{chngcntr}
\usepackage{url}
\usepackage[compatibility=false, small, hang, nooneline]{caption}
\usepackage{subcaption}
\usepackage[boxruled, lined, linesnumbered]{algorithm2e}
\usepackage[colorinlistoftodos, textwidth=3.75cm,]{todonotes}

\expandafter\ifx\csname DontPrintSemicolon\endcsname\relax
\newcommand{\DontPrintSemicolon}{\dontprintsemicolon}\fi
\newcommand{\plll}{PotLLL\xspace}
\newcommand{\dlll}{DeepLLL\xspace}
\newcommand{\ntl}{NTL\xspace}

\newcommand{\tL}{\mathcal{L}}

\newcommand{\pot}{\mathrm{Pot}}

\newcommand{\argmin}{\mathrm{argmin}}

\newcommand{\R}{\mathbb{R}}
\newcommand{\Q}{\mathbb{Q}}
\newcommand{\Z}{\mathbb{Z}}

\newcommand{\spa}{\mathrm{span}\,}

\newcommand{\vol}{\mathrm{vol}\,}

\newcommand{\norm}[1]{\lVert #1 \rVert}

\newtheorem{definition}{Definition}[section]

\newtheorem{proposition}[definition]{Proposition}
\newtheorem{lemma}[definition]{Lemma}
\newtheorem{corollary}[definition]{Corollary}

\newcommand{\grossO}[1]{\ensuremath{\mathcal{O}(#1)}}

\begin{document}

\title{A Polynomial Time Version of LLL With Deep Insertions}

\date{\today}

\author{Felix Fontein\footnote{Universit\"at Z\"urich, \href{mailto:felix.fontein@math.uzh.ch}{felix.fontein@math.uzh.ch}} 
   \and Michael Schneider\footnote{Technische Universit\"at Darmstadt, \href{mailto:mischnei@cdc.informatik.tu-darmstadt.de}{mischnei@cdc.informatik.tu-darmstadt.de}} 
   \and Urs Wagner\footnote{Universit\"at Z\"urich, \href{mailto:urs.wagner@math.uzh.ch}{urs.wagner@math.uzh.ch}}}

\maketitle

\begin{abstract}

Lattice reduction algorithms have numerous applications in number theory, algebra, as well as in cryptanalysis.
The most famous algorithm for lattice reduction is the LLL algorithm. In polynomial time it computes a reduced basis with provable output quality.
One early improvement of the LLL algorithm was LLL with deep insertions (\dlll).
The output of this version of LLL has higher quality in practice but the running
time seems to explode. Weaker variants of \dlll, where the insertions are
restricted to blocks, behave nicely in practice concerning the running time.
However no proof of polynomial running time is known.
In this paper a new variant of \dlll with provably polynomial running time is presented. We compare the practical behavior of the new algorithm to classical LLL, BKZ as well as blockwise variants of \dlll regarding both the output quality and running time.

\vspace{3mm}

\end{abstract}

\section{Introduction}\label{sec:introduction}

The well-known LLL lattice reduction algorithm was presented in 1982 by Lenstra, Lenstra, Lov\'asz \cite{LLL}. 
Apart from various other applications (e.g.~\cite[Chapter 9,10]{ng10}) it has already at an early stage been used to attack various public key cryptosystems.
Nevertheless lattice problems remain popular when it comes to the construction of provably secure cryptosystems (e.g.~\cite[Chapter~13]{ng10}). 
Consequently  improvements in lattice reduction still have a direct impact on the security of many cryptosystems and rise high interest in the crypto-community.

Many lattice reduction algorithms used in practice are generalizations of the LLL algorithm. 
The Block-Korkine-Zolotarev (BKZ) reduction algorithm by Schnorr and Euchner~\cite{BKZ} is probably the most used algorithm when stronger reduction than the one achieved by LLL is required. It can be seen as a generalization of LLL to higher blocksizes, and while the running time seems to behave well for small blocksizes~\cite{gam08}, no useful upper bound has been proven so far.
Another improvement of the LLL algorithm has also been suggested in~\cite{BKZ}. 
While in LLL adjacent basis vectors are swapped if certain conditions are satisfied, in the so called  LLL with deep insertions (\dlll in the sequel), basis vectors can be swapped even when not adjacent. The practical behavior of~\dlll when it comes to the reducedness of the output basis is superior the one of LLL. Unfortunately also the running time explodes and does not seem to be polynomial in the dimension of the lattice. One attempt get across this problem, is to restrict the insertions to certain blocks of basis vectors. While the authors in \cite{BKZ} claim that these blockwise restriction variants of~\dlll run in polynomial time, we are not aware of any proof thereof. 
For an overview on the practical behavior of the different variants and improvements on LLL, we refer to~\cite{ng06,gam08}.
 There the practical behavior of the reduction algorithms is investigated using
the widely used~\ntl library.

In this paper we present a new version of~\dlll, called~\plll. To our knowledge
it is the first improvement of LLL with regard to deep insertions which provably runs in polynomial time. The
practical behavior of~\plll regarding both the output quality and running time
is empirically tested and compared to BKZ and~\dlll with different blocksizes. 
The tests are performed with a completely new implementation of the different
reduction algorithms. This additionally allows an independent review of the
results in~\cite{ng06,gam08}. The tests indicate that our algorithm can serve as a
serious alternative to BKZ with low blocksizes.

The paper is organized as follows. In Section~\ref{sec:prelim} all
necessary notations and definitions are given.
In Section~\ref{sec:algorithm} the reduction notion and the new algorithm is
presented and a theoretical analysis is provided. 
Section~\ref{sec:experiments} contains the empirical results and conclusions are drawn in Section~\ref{sec:conclusion}.

\section{Preliminaries}
\label{sec:prelim}

 A lattice $\tL \subset \R^m$ of rank $n$ and dimension $m$ is a discrete subgroup of $\R^m$
generated by integer
    linear combinations of $n$ linearly independent vectors $b_1,\dots,b_n$ in
$\R^m$:
    \[
    \tL= \tL(b_1,\dots,b_n) := \biggl\{\sum_{i=1}^n x_ib_i \biggm| \forall i :
x_i \in \Z \biggr\}\,.
    \]
We will often write the basis $b_1,\dots,b_n$ as rows of a matrix $B$ in the following way $B= [b_1,\dots,b_n]$.
In order to have exact representations in computers, only lattices in $\Q^n$ are
considered. Simple scaling by the least common multiple of the denominators
allows us to restrict ourselves to integer lattices $\tL \subseteq \Z^m$.
The volume of a lattice $\tL(B)$ equals the volume of its fundamental
parallelepiped $\vol(\tL)=\sqrt{\det(BB^t)}$.
For $n\geq 2$, a lattice has infinitely many bases as $\tL(B)=\tL(B')$ if and only if $\exists U \in
GL_n(\Z): B=UB'$. Therefore, the volume of a lattice is well defined.
By $\pi_k: \R^m \rightarrow \spa\{b_1, \dots, b_{k-1} \}^\bot$ we denote the
orthogonal projection from $\R^m$
  onto the orthogonal complement of $\spa\{b_1, \dots, b_{k-1} \}$.
 In
particular, $\pi_1 =
  \mathrm{id}_{\R^m}$ and $b^*_i:=\pi_i(b_i)$ equals the $i$-th basis vector of the Gram-Schmidt orthogonalization 
$B^*=[b^*_1,\dots, b^*_n]$ of $B$.
By $\mu_{i,j}:= \langle b_i,b^*_j \rangle / \langle b^*_j,b^*_j \rangle$,
$j<i$, we denote the Gram-Schmidt coefficients. The Gram-Schmidt vectors can
iteratively be computed by
 $\pi_i(b_i)=b_i^* = b_i - \sum_{j=1}^{i-1} \mu_{i,j} b_j^*$.

Throughout this paper, by $\norm{\cdot}$ we denote the Euclidean norm and by
$\lambda_1(\tL)$ we denote the length of a shortest non-zero vector in $\tL$
with respect to the Euclidean norm: $\lambda_1(\tL):=\min_{v\in \tL} \norm{v}$.
Determining $\lambda_1(\tL)$ is commonly known as the shortest vector problem
(SVP) and is proven to be NP-hard (under randomized reductions) (see e.g.
\cite{mi02}). 
Upper bounds with respect to the determinant exist, for all rank $n$ lattices $\tL$ we have  \cite{ng10}
\[
\frac{\lambda_1(\tL)^2}{\vol(\tL)^{2/n}}\leq \gamma_n \leq 1+\frac{n}{4}\,,
\]
where $\gamma_n$ is the \emph{Hermite constant} in dimension $n$.
Given a relatively short vector $v \in \tL$, one measures its quality by the
\emph{Hermite factor} $\norm{v}/\vol(\tL)^{1/n}$ it achieves.
Modern lattice reduction algorithms achieve a Hermite factor which is
exponential in $n$ and no polynomial time algorithm is known to achieve linear
or polynomial Hermite factors.

 Let $S_n$ denote the group of permutations of $n$ elements. By applying $\sigma \in S_n$ to a basis
 $B=[b_1,\dots,b_n]$, the basis vectors are reordered $\sigma
 B=[b_{\sigma(1)},\dots,b_{\sigma(n)}]$. For $1\leq k \leq \ell \leq n$ we define a class of
 elements $\sigma_{k,\ell} \in S_n$ as follows:
\begin{equation}
\sigma_{k,\ell}(i)=\left\{
\begin{array}{lll}
i & \mbox{for} &  i<k \mbox{ or } i>\ell\,, \\
\ell & \mbox{for} & i=k\,, \\
i-1 & \mbox{for} & k<i\leq \ell\,.
\end{array}
\right.
\end{equation}
Note that $\sigma_{k,\ell} = \sigma_{k,k+1} \sigma_{k+1,k+2} \cdots \sigma_{\ell-1,\ell}$ and that
$\sigma_{k,k+1}$ is swapping the two elements~$k$ and $k + 1$.

  \begin{definition}
    Let $\delta \in \left(1/4, 1\right]$.  A basis $B=[b_1,\dots,b_n]$
of a lattice
    $\tL(b_1,\dots,b_n)$ is called \emph{$\delta$-LLL reduced} if and only if it
satisfies the
    following two conditions:
    \begin{enumerate}
      \item $\forall 1\leq j<i\leq n: |\mu_{i,j}| \leq
\frac{1}{2}$ (size-reduced).
      \item  $1\leq k <n : \delta \cdot \|\pi_k(b_k)\|^2 \leq
\|\pi_k(b_{k+1})\|^2$ (Lov{\'a}sz-condition).
    \end{enumerate}
  \end{definition}
  A $\delta$-LLL reduced basis $B=[b_1,\dots,b_n]$ can be computed in polynomial time \cite{LLL} and  provably satisfies the following
bounds:
\begin{equation}
\|b_1\| \leq \bigl(\delta-1/4\bigr)^{-(n-1)/2} \cdot \lambda_1(\tL(B)) \quad \text{and} \quad
\|b_1\| \leq \bigl(\delta-1/4\bigr)^{-(n-1)/4} \cdot \vol(\tL(B))^{1/n} \label{equ:hermiteLLL}.
\end{equation}
  While these bounds can be reached, they are worst case bounds. In practice,
 LLL reduction algorithms behave much better \cite{ng06}.
One early attempt to improve the LLL reduction algorithm is due to Schnorr and
Euchner \cite{BKZ} who came up with the notion of a \dlll reduced basis:
\begin{definition}
    Let $\delta \in \left(1/4, 1\right]$.  A basis $B=[b_1,\dots,b_n]$
of a lattice
    $\tL(b_1,\dots,b_n)$ is called \emph{$\delta$-\dlll reduced with blocksize $\beta$} if and only if
it satisfies the
    following two conditions:
    \begin{enumerate}
      \item $\forall 1\leq j<i\leq n: |\mu_{i,j}| \leq \frac{1}{2}$ 
(size-reduced).
      \item  $\forall 1\leq k < \ell \leq n \mbox{ with } k\leq \beta \vee k-\ell \leq \beta: \delta \cdot \|\pi_k(b_k)\|^2 \leq\|\pi_k(b_{\ell})\|^2$.
    \end{enumerate}
\end{definition}
If $\beta=n$ we simply call this a \dlll reduced basis.
While the first basis vector of \dlll reduced bases in the worst case does not
achieve a better Hermite factor than classical LLL (see Section
\ref{sec:critical}), the according reduction algorithms usually return much
shorter vectors than pure LLL.
Unfortunately no polynomial time algorithm to compute \dlll reduced bases is known. 

The following definition is used in the proof (see e.g. \cite{mi02}) of the polynomial running time of the LLL reduction algorithm and will play a main role in our improved variant of LLL.
\begin{definition}
The \emph{potential} $\pot(B)$ of a lattice basis $B=[b_1,\dots,b_n]$ is defined as
\[
\pot(B):=\prod^n_{i=1} \vol(\tL(b_1,\dots,b_i))^2=\prod^n_{i=1}
\|b^*_i\|^{2(n-i+1)}\,.
\]
\end{definition}
Here it is used that $\vol(\tL) = \prod_{i=1}^{n} \norm{b_i^*}$.
Note that, unlike the volume of the lattice, the potential of a basis is variant
under basis permutations. The following lemma describes how the potential
changes if $\sigma_{k,\ell}$ is applied to the basis.

\begin{lemma}
\label{lem:pot}
Let $B=[b_1,\dots,b_n]$ be a lattice basis. Then for $1\leq k \leq \ell \leq n$
\[
\pot(\sigma_{k,\ell}B)=\pot(B) \cdot \prod^{\ell}_{i=k} \frac{\|\pi_i(b_{\ell})\|^2}{\|\pi_i(b_i)\|^2}.
\]
\end{lemma}
\begin{proof}
  First note that it is well-known that $\pot(\sigma_{k,k+1}B)= \|\pi_k(b_{k+1})\|^2 /
  \|\pi_k(b_k)\|^2 \cdot \pot(B)$. This property is used in the proofs of the polynomial running
  time of LLL \cite{LLL,mi02}.

  We prove the claim by induction over $\ell-k$.  The claim is true for $k=\ell$. For $k<\ell$,
  $\sigma_{k,\ell}=\sigma_{k,k+1}\sigma_{k+1,\ell}$. As $b_\ell$ is the $(k+1)$-th basis vector of
  $\sigma_{k+1,\ell}B$, with the above identity we get
  $\pot(\sigma_{k,\ell}B)=\pot(\sigma_{k,k+1}\sigma_{k+1,\ell}B)=\frac{\|\pi_k(b_{\ell})\|^2}{\|\pi_k(b_k)\|^2}
  \cdot \pot(\sigma_{k+1,\ell}B)$, which completes the proof. \qed
\end{proof}

\section{The Potential-LLL Reduction}
\label{sec:algorithm}

In this section we present our polynomial time variant of \dlll. We start
with the definition of a $\delta$-\plll\ reduced basis. Then we present an
algorithm that outputs such a basis followed by a runtime proof. 
\begin{definition}
\label{def:potLLL}
Let $\delta \in (1/4,1]$. A lattice basis $B=[b_1,\dots,b_n]$ is \emph{$\delta$-\plll reduced} if and only if
\begin{enumerate}
\item $\forall 1\leq j<i\leq n: |\mu_{i,j}| \leq \frac{1}{2}$ (size-reduced).
\item $\forall 1\leq k<\ell\leq n: \delta \cdot \pot(B) \leq
\pot(\sigma_{k,\ell}B)$.
\end{enumerate}

\end{definition}

\begin{lemma}
A $\delta$-\plll reduced basis $B$ is also $\delta$-LLL reduced. 
\end{lemma}
\begin{proof}
  Lemma~\ref{lem:pot} shows that $\delta \cdot \pot(B) \leq \pot(\sigma_{i,i+1}B)$ if and only if
  $\delta \|\pi_i(b_i)\|^2 \leq\|\pi_i(b_{i+1})\|^2$. Thus the Lov\'asz condition is implied by the
  second condition in Definition~\ref{def:potLLL} restricted to consecutive pairs, i.e.\ $\ell = k +
  1$. \qed
\end{proof}

\begin{lemma}
\label{lem:dlllplll}
A $\delta$-\dlll reduced basis $B$ is also $\delta^{n-1}$-\plll reduced.
\end{lemma}
\begin{proof}
  We proceed by contradiction.  Assume that $B$ is not $\delta^{n-1}$-\plll reduced, i.e. there
  exist $1\leq k < \ell \leq n$ such that $\delta^{n-1}\pot(B)> \pot(\sigma_{k,\ell}B)$. By
  Lemma~\ref{lem:pot} this is equivalent to
  \[
  \delta^{n-1} > \prod^\ell_{i=k}
  \frac{\|\pi_i(b_\ell)\|^2}{\|\pi_i(b_i)\|^2}=\prod^{\ell-1}_{i=k}
  \frac{\|\pi_i(b_\ell)\|^2}{\|\pi_i(b_i)\|^2}\,.
  \]
  It follows that there exist a $j \in [k, \ell-1]$ such that
  $\|\pi_j(b_\ell)\|^2/\|\pi_j(b_j)\|^2<\delta^{(n-1)/(\ell-k)}\leq\delta$ which implies that $B$ is
  not $\delta$-\dlll reduced. \qed
\end{proof}

\subsection{High-Level Description}

A high-level version of the algorithm is presented as Algorithm~\ref{alg:potLLL}. The algorithm is
very similar to the classical LLL algorithm and the classical \dlll reduction by Schnorr and Euchner
\cite{BKZ}. During its execution, the first $\ell-1$~basis vectors are always $\delta$-\plll
reduced (this guarantees termination of the algorithm). As opposed to classical
LLL, and similar to \dlll, $\ell$ might decrease by more than
one. This happens precisely during deep insertions: in these cases, the $\ell$-th vector is not
swapped with the $(\ell-1)$-th one, as in classical LLL, but with the $k$-th one for $k < \ell -
1$. In case $k = \ell - 1$, this equals the swapping of adjacent basis vectors as in classical LLL.
The main difference of \plll\ and \dlll\ is the condition that controls
insertion of a vector.

\begin{algorithm}[htbp]
  \caption{Potential LLL}
  \label{alg:potLLL}
  \SetKwComment{Comment}{$\vartriangleright$~}{}
  \SetCommentSty{textit}
  \DontPrintSemicolon
  
  \KwIn{Basis $B \in \Z^{n \times m}$, $\delta \in (1/4,1]$}
  \KwOut{A $\delta$-\plll reduced basis.}
  $\ell \leftarrow 1$\;
  \While{$\ell \leq n$}{
    	Size-reduce$(B)$\;\label{alg:potLLL:sizereduce}
	$k \leftarrow \argmin_{1\leq j\leq \ell}\pot(\sigma_{j,\ell}B)$\label{alg:potLLL:min}\;	
	\eIf{$\delta \cdot \pot(B)>\pot(\sigma_{k,\ell}B$)\label{alg:potLLL:if}}{
			$B \leftarrow \sigma_{k,\ell}B$\;
			$\ell \leftarrow k$ \;
	}{
		$\ell \leftarrow \ell+1$ \;
	}
  }
\Return $B$\;
\end{algorithm}

\subsection{Detailed Description}
There are two details to consider when implementing Algorithm~\ref{alg:potLLL}. The first one is
that since the basis vectors~$b_1, \dots, b_{\ell-1}$ are already $\delta$-\plll reduced, they are
in particular also size-reduced. Moreover, the basis vectors $b_{\ell+1}, \dots, b_n$ will be
considered later again. So in line~\ref{alg:potLLL:sizereduce} of the algorithm
it suffices to
size-reduce $b_\ell$ by $b_1, \dots, b_{\ell-1}$ as in classical LLL. Upon termination, when $\ell =
n + 1$, the whole basis will be size-reduced.

Another thing to consider is the computation of the potentials of $B$ and $\sigma_{j,\ell} B$ for $1
\le j \le \ell$ in line~\ref{alg:potLLL:min}.  Computing the potential of the basis is a rather slow
operation. But we do not need to compute the potential itself, but only compare
$\pot(\sigma_{k,\ell} B)$ to $\pot(B)$; by Lemma~\ref{lem:pot}, this quotient can be efficiently
computed. Define $P_{k,\ell} := \pot(\sigma_{k,\ell} B) / \pot(B)$. The ``if''-condition in
line~\ref{alg:potLLL:if} will then change to $\delta > P_{k,\ell}$, and the minimum in
line~\ref{alg:potLLL:min} will change to $\argmin_{1 \leq j \leq \ell} P_{j,\ell}$. Using
$P_{\ell,\ell} = 1$ and
\begin{equation}
\label{equ:P_jl}
P_{j,\ell}= \frac{\pot(\sigma_{j,\ell}
B)}{\pot(B)}=P_{j+1,\ell} \cdot \frac{\|\pi_{j}(b_\ell)\|^2}{\|\pi_{j}(b_{j})\|^2} =
P_{j+1,\ell} \cdot \frac{\|b_\ell^*\|^2+ \sum_{i=j}^{\ell-1} \mu_{\ell,i}^2
\|b_i^*\|^2}{\|b_{j}^*\|^2}
\end{equation}
for $j < \ell$ (Lemma~\ref{lem:pot}), we can quickly determine $\argmin_{1 \leq j \leq \ell}
P_{j,\ell}$ and check whether $\delta > P_{k,\ell}$ if $j$ minimizes $P_{j,\ell}$.

 A detailed version of Algorithm~\ref{alg:potLLL} with these steps
filled in is described as Algorithm~\ref{alg:potLLL:full}.
On line~\ref{alg:potLLL:full:potmod} of Algorithm~\ref{alg:potLLL:full}, $P_{j,\ell}$ is iteratively computed as in
Equation~\eqref{equ:P_jl}.  Clearly, the algorithm could be further improved by iteratively
computing $\|\pi_j(b_\ell)\|^2$ from $\|\pi_{j+1}(b_\ell)\|^2$.  Depending on the implementation of
the Gram-Schmidt orthogonalization, this might already have been computed and stored. For example,
when using the Gram-Schmidt orthogonalization as described in Figure~4
of~\cite{nguyen-stehle-fplll-revisited}, then $\|\pi_j(b_\ell)\|^2 = s_{j-1}$ after computation of
$\|b_\ell^*\|^2$ and $\mu_{\ell,j}$ for $1 \le j < \ell$.

\begin{algorithm}[htbp]
  \caption{Potential LLL, detailed version}
  \label{alg:potLLL:full}
  \SetKwComment{Comment}{$\vartriangleright$~}{}
  \SetCommentSty{textit}
  \DontPrintSemicolon
  
  \KwIn{Basis $B \in \Z^{n \times m}$, $\delta \in (1/4,1]$}
  \KwOut{A $\delta$-\plll reduced basis.}
  $\ell \leftarrow 1$\;
  \While{$\ell \leq n$\label{alg:potLLL:full:while}}{
        Size-reduce$(b_\ell \text{ by } b_1, \dots, b_{\ell-1})$\;\label{alg:potLLL:full:sizereduce}
        Update($\|b_\ell^*\|^2$ and $\mu_{\ell,j}$ for $1 \le j <
\ell$)\;\label{alg:potLLL:full:update1}
        $P \leftarrow 1$, \quad $P_{\min} \leftarrow 1$, \quad $k \leftarrow 1$ \;
        \For{$j=\ell-1$ down to $1$}{
                $P \leftarrow P \cdot \frac{ \|b_\ell^*\|^2 +\sum_{i=j}^{\ell-1} \mu_{\ell,i}^2 \|b_i^*\|^2}{\|b_j^*\|^2}$\label{alg:potLLL:full:potmod} \;
                \If{$P < P_{\min}$}{
                        $k \leftarrow j$ \;
                        $P_{\min} \leftarrow P$ \;
                }
        }
        \eIf{$\delta > P_{\min}$\label{alg:potLLL:full:if}}{
                $B \leftarrow \sigma_{k,\ell}B$\;
                Update($\|b_k^*\|^2$ and $\mu_{k,j}$ for $1 \le j <
k$)\;\label{alg:potLLL:full:update2}
                $\ell \leftarrow k$ \;
        }{
                $\ell \leftarrow \ell+1$ \;
        }
  }
\Return $B$\;
\end{algorithm}

\subsection{Complexity Analysis}

Here we show that the number of operations in the \plll algorithm are
bounded polynomially in the dimension $n$ and the logarithm of the input size.
We present the runtime for Algorithm~\ref{alg:potLLL:full}.

\begin{proposition}
\label{prop:potLLL2}
Let $\delta \in (1/4, 1)$ and $C = \max_{i=1\ldots n} \norm{b_i}^2$. Then
Algorithm~\ref{alg:potLLL:full} performs $\grossO{n^3 \log_{1/\delta}(C)}$ iterations of the
\texttt{while} loop in line~\ref{alg:potLLL:full:while} and a total of $\grossO{(n + m) n^4
  \log_{1/\delta}(C)}$ arithmetic operations.
\end{proposition}
\begin{proof}
Let us start by upper bounding the potential $I$ of the input basis with respect to $C$. 
Let
$d_j := \vol\left(\tL(b_1,\dots,b_j)\right)^2 = \prod_{i=1}^j \norm{b_i^*}^2$ for $j=1,\dots,n$.
Recall that $\norm{b_i^*}^2 \leq \norm{b_i}^2 \leq C$ for $i=1,\dots,n$ and hence $d_j<C^j$. Consequently we have the following upper bound on the potential
\begin{equation}
\label{equ:I}
I = \prod_{j=1}^{n-1} d_j \cdot \vol(\tL) \leq \prod_{j=1}^{n-1} C^j \cdot \vol(\tL) \leq
C^{\frac{n(n-1)}{2}} \cdot \vol(\tL)\,.
\end{equation}
Now, by a standard argument, we show that the number of iterations of the
while loop is bounded by $\grossO{n^3 \log_{1/\delta}(C)}$.
In each iteration, either the iteration counter~$\ell$ is increased by 1, or a swapping takes place
and $\ell$ is decreased by at most $n - 1$. In the
swapping case, the potential is decreased by a factor at least $\delta$. So after $N$ swaps the potential $I_N$ satisfies $I \geq (1/\delta)^NI_N \geq (1/\delta)^N \cdot \vol(\tL)$ using the fact that $I_N \geq \vol(\tL)$. Consequently the number of swaps~$N$ is
bounded by $N \leq \log_{1/\delta} (I / \vol(\tL))$.
By Equation (\ref{equ:I}) we get that $N \leq \log_{1/\delta}\bigl(C^{n(n-1)/2}\bigr)$.
Now note that
the number~$M$ of iterations where~$\ell$ is increased by 1 is at most $M \le (n - 1) \cdot N + n$.
This shows that the number of iterations is bounded by $\grossO{n^3
\log_{1/\delta}(C)}$.

Next we show that the number of operations performed in each iteration of the \texttt{while} loop is
dominated by $\grossO{n^2 + n m}$ operations.
Size-reduction (line~\ref{alg:potLLL:full:sizereduce}) and the first update
step (line~\ref{alg:potLLL:full:update1}) can be done in \grossO{n m} steps. The
for-loop consists of \grossO{n} iterations where the most expensive operation
is the update of P in line~\ref{alg:potLLL:full:potmod}. Therefore the loop
requires \grossO{n^2} arithmetic operations. Swapping can be done in \grossO{n}
operations, whereas the second update in line~\ref{alg:potLLL:full:update2}
requires again \grossO{n^2} operations.

It follows that each iteration costs at most \grossO{n^2 + n m} arithmetic operations. This shows
that in total the algorithm performs $\grossO{(n + m) n^4 \log(C)}$ operations. \qed
\end{proof}

\subsection{Worst-Case Behavior}
\label{sec:critical}
For $\delta=1$, there exist so called \emph{critical bases} which are
$\delta$-LLL reduced bases and whose Hermite factor reaches the worst case bound
in~(\ref{equ:hermiteLLL})~\cite{sc94}.
These bases can be adapted to form a \dlll reduced basis where the first vector
reaches the worst case bound  in~(\ref{equ:hermiteLLL}).

\begin{proposition}\label{prop:critical}
For $\alpha = \sqrt{3/4}$,  the rows of $B = A_n(\alpha)$ (see below) define a $\delta$-\dlll reduced basis with $\delta=1$ and $\|b_1\|^2=\frac{1}{\alpha^{(n-1)/2}} \vol(\tL(A_n))^{1/n}$.
\end{proposition}

\begin{equation}
A_n(\alpha):=\left(
\begin{matrix}
1			&	0				& 	 \cdots			& \cdots 		&		\cdots			&0\\
\frac{1}{2}	& 	\;\alpha 			&	\ddots
&			&						&
\vdots\\
\vdots		&	\;\frac{\alpha}{2}	& \;\alpha^2
& \ddots		&						& \vdots \\
 \vdots 	&	\vdots 			& \frac{\alpha^2}{2}	&\ddots		&	\ddots				&  \vdots \\
 \vdots		&	\vdots			&\vdots 				&\ddots		&\alpha^{n-2}				& 0 \\
\frac{1}{2}	&	\;\frac{\alpha}{2}	&\;\frac{\alpha^2}{2}	& \hdots
		& \;\frac{\alpha^{n-2}}{2}	& \;\alpha^{n-1}
\end{matrix}
\right)
\end{equation}

\begin{proof}
From the diagonal form of $A_n$ it is easy to see that $\vol(\tL)=\det(A_n)=\alpha^{n(n-1)/2}$.
Hence $\|b_1\|^2=1=1/\alpha^{(n-1)/2}\vol(\tL)$. 
It remains to show that $A_n$ is \dlll reduced. Note that $B^*$ is a diagonal
matrix with the same entries on the diagonal as $B$.
Note that it is size reduced as for all $1\leq j<i\leq n$ we have
$\mu_{i,j}=\langle b_i,b^*_j\rangle/\langle b^*_j,b^*_j
\rangle=\tfrac{1}{2} \alpha^{2 (j-1)}/\alpha^{2 (j-1)}=\frac{1}{2}$.
Further, using that $\pi_j(b_i)=b_i^* +
\sum_{\ell=j}^{i-1}\mu_{i,\ell}b_{\ell}^*$, we have that
\[
\|\pi_j(b_i)\|^2=\alpha^{2(i-1)}+\frac{1}{4}\sum^{i-1}_{\ell=j}
\alpha^{2(\ell-1)}=\alpha^{2(j-1)}\left(\frac{1}{4}\sum^{i-j-1}_{\ell=0}
\alpha^{2\ell}+\alpha^{2(i-j)}\right)\;.
\]
As for $\alpha=\sqrt{3/4}$, we have that
$\frac{1}{4}\sum^{i-j-1}_{\ell=0} \alpha^{2\ell}+\alpha^{2(i-j)}=1$,
and hence $\|\pi_j(b_i)\|^2=\alpha^{2(j-1)}= \|\pi_j(b_j)\|^2$. Therefore, the
norms of the projections for fixed~$j$ are all equal, and $A_n(\alpha)$ is $\delta$-\plll-reduced with
$\delta=1$. \qed
\end{proof}
Using Lemma~\ref{lem:dlllplll}, we obtain:
\begin{corollary}
For $\alpha = \sqrt{3/4}$,  the rows of $A_n(\alpha)$ define a $\delta$-\plll reduced basis
with $\delta=1$ and $\|b_1\|^2=\frac{1}{\alpha^{(n-1)/2}} \vol(\tL(A_n))^{1/n}$. \qed
\end{corollary}

\section{Experimental Results}\label{sec:experiments}
Extensive experiments have been made to examine how the classical LLL reduction algorithm performs
in practice \cite{ng06,gam08}. We ran some experiments to compare our \plll
algorithm to our implementations of LLL, \dlll, and BKZ.

\subsection{Setting}
We run the following algorithms, each with the standard reduction parameter $\delta=0.99$:
\begin{enumerate}
\item classical LLL,
\item \plll,
\item \dlll with blocksize $\beta=5$ and $\beta=10$ (the latter up to dimension
240 only),
\item BKZ with blocksize $5$ (BKZ-5) and $10$ (BKZ-10).
\end{enumerate}
The implementations all use the same arithmetic back-end. Integer arithmetic is done using GMP, and
Gram-Schmidt arithmetic is done as described in~\cite[Figures~4 and
5]{nguyen-stehle-fplll-revisited}. As floating point types, \texttt{long double} (x64 extended
precision format, 80~bit representation) and MPFR arbitrary precision floating point numbers are
used with a precision as described in \cite{nguyen-stehle-fplll-revisited}. The implementations of
\dlll and BKZ follow the classical description in \cite{BKZ}. \plll was implemented as described in
Algorithm~\ref{alg:potLLL:full} (page~\pageref{alg:potLLL:full}).

We ran experiments in dimensions 40 to 300, considering the dimensions which are multiples of~10 from 40 to 200 and dimensions that are multiples of~20 from 200 to 300. In each dimension, we considered 50 random lattices. More
precisely, we used the lattices of seed 0 to 49 from the SVP
Challenge\footnote{\url{http://www.latticechallenge.org/svp-challenge}}. For each lattice, we used
two bases: the original basis and a $0.75$-LLL reduced basis.

All experiments were run on Intel\textsuperscript{\textregistered}
Xeon\textsuperscript{\textregistered} X7550 CPUs at 2~GHz on a shared memory machine. For dimensions
40 up to 160, we used \texttt{long double} arithmetic, and for dimensions 160 up to 300, we used
MPFR. In dimension 160, we did the experiments both using \texttt{long double} and MPFR
arithmetic. The reduced lattices did not differ. In dimension~170, floating point errors prevented
the \texttt{long double} arithmetic variant to complete on some of the lattices.

\subsection{Results}
For each run, we recorded the length of the shortest vector as well as the required CPU time for the
reduction. Our main interest lies in the $n$-th root of the \emph{Hermite factor}
$\frac{\|b_1\|}{\vol(\tL)^{1/n}}$, where $b_1$ is the shortest vector of the basis of $\tL$
returned. Figure~\ref{fig:overview} (see pages~\pageref{fig:overview}-\pageref{fig:comparism:pp:ld}
for all figures) compares the average $n$-th root of the Hermite factor and average logarithmic
running time of the algorithms for all dimensions. The graphs also show confidence intervals for the
average value with a confidence level of 99.9\%.

As one can see, there is a clear hierarchy with respect to the achieved Hermite factor. Our \plll
performs better than BKZ-5, though worse than \dlll with $\beta = 5$ and BKZ-10, which in turn
perform worse than \dlll with $\beta = 10$. The behavior for preprocessed bases and bases in Hermite
normal form is very similar. We collected the average $n$-th root Hermite factors $\|b_1\|^{1/n}
\cdot \vol(\tL)^{-1/n^2}$ in Table~\ref{table:red} and compared them to the worst-case bound in
Equation~\eqref{equ:hermiteLLL}. Our data for LLL is similar to the one in \cite{ng06} and
\cite[Table~1]{gam08}. However, we do not see convergence of the $n$-th root Hermite factors in our
experiments, as they are still increasing even in high dimensions $n>200$.

\begin{table}[h]
  \renewcommand{\arraystretch}{1.3}
  \begin{center}
    \begin{tabular}{  l l l l l l l }
      \toprule
      Dimension & & $n = 100$ & $n = 160$ & $n = 220$ & $n = 300$ \\
      \midrule
      Worst-case bound (proven) & ~\ & $\approx 1.0774$ & $\approx 1.0777$ & $\approx 1.0778$ & $\approx 1.0779$ \\
      Empirical $0.99$-LLL & & $ 1.0187$ & $ 1.0201$ & $ 1.0207$ & $ 1.0212$ \\ 
      Empirical $0.99$-BKZ-5 & ~\ & $ 1.0154$ & $ 1.0158$ & $ 1.0161$ & $ 1.0163$ \\ 
      Empirical $0.99$-\plll & & $ 1.0146$ & $ 1.0150$ & $ 1.0152$ & $ 1.0153$ \\ 
      Empirical $0.99$-\dlll with $\beta=5$ & ~\ & $ 1.0138$ & $ 1.0142$ & $ 1.0147$ & $ 1.0150$ \\ 
      Empirical $0.99$-BKZ-10 & ~\ & $ 1.0140$ & $ 1.0143$ & $ 1.0144$ & $ 1.0145$ \\ 
      Empirical $0.99$-\dlll with $\beta=10$ & ~\ & $ 1.0128$ & $ 1.0132$ & $ 1.0135$ & \;\quad--- \\ 
      \toprule
    \end{tabular}
  \end{center}
  \caption{Worst case bound and average case estimate for $\delta$-LLL reduction, $\delta$-\dlll
  reduction, $\delta$-\plll reduction and $\delta$-BKZ reduction of the $n$-th root Hermite factor
  $\|b_1\|^{1/n} \cdot \vol(\tL)^{-1/n^2}$. The entries are sorted in descending order with respect
  to the observed Hermite factors.}
  \label{table:red}
\end{table}

For the running time comparison, Figure~\ref{fig:overview:time} shows that the observed order is
similar to the order induced by the Hermite factors. LLL is fastest, followed by BKZ-5 and \plll,
then by BKZ-10 and \dlll with $\beta = 5$, and finally there is \dlll with $\beta = 10$. The running
time of \plll and BKZ-5 is very close to each other for higher dimensions, while \plll is clearly
slower for smaller dimensions. While Figure~\ref{fig:overview:time} shows that BKZ-5 is usually
slightly faster than \plll and BKZ-10 slightly faster than \dlll with $\beta = 5$, the behavior is
more interesting if one considers preprocessed and non-preprocessed bases separately. We do this in
Figures~\ref{fig:comparism:up} and \ref{fig:comparism:pp}. Recall that the unprocessed bases are
bases in Hermite normal form, and the processed bases are the same bases run through 0.75-LLL.

In Figures~\ref{fig:comparism:up}, we compare the behavior for unprocessed bases in Hermite normal
form. Every line connecting bullets corresponds to the behavior of one algorithm for different
dimensions. Again, the box surrounding a bullet is a confidence interval with confidence level
99.9\%. The shaded regions show which Hermite factors can be achieved in every
dimension
by these algorithms. Algorithms on the border of the region are optimal for their Hermite factor:
none of the other algorithms in this list produces a better average Hermite factor in less time. In
Figure~\ref{fig:comparism:up:ld}, one can see that BKZ-5 produces worse output slower than \plll up
to dimension~160. Also, BKZ-10 is inferior to \dlll with $\beta = 5$ as it is both slower and
produces worse Hermite factors. As the dimension increases, the difference in running time becomes
less and less. In fact, for dimension~180 and larger, BKZ-10 becomes faster than \dlll with $\beta =
5$ (Figure~\ref{fig:comparism:up:real}).

On the other hand, for preprocessed bases, the behavior is different, as
Figure~\ref{fig:comparism:pp} shows. Here, BKZ-5 is clearly faster than \plll and BKZ-10 clearly
faster than \dlll with $\beta = 5$. In fact, for dimensions 60, 80 and 100, \plll is slower than
BKZ-10 while producing worse output (Figure~\ref{fig:comparism:pp:ld}). For higher dimensions, \plll
is again faster than BKZ-10 (Figure~\ref{fig:comparism:pp:real}), though not
substantially. Therefore, for preprocessed bases, it seems that BKZ-10 is more useful than \plll
and \dlll with $\beta = 5$.

\begin{figure}[h]
  \begin{center}
    \begin{subfigure}[t]{\textwidth}
      \begin{center}
        \includegraphics[width=\textwidth]{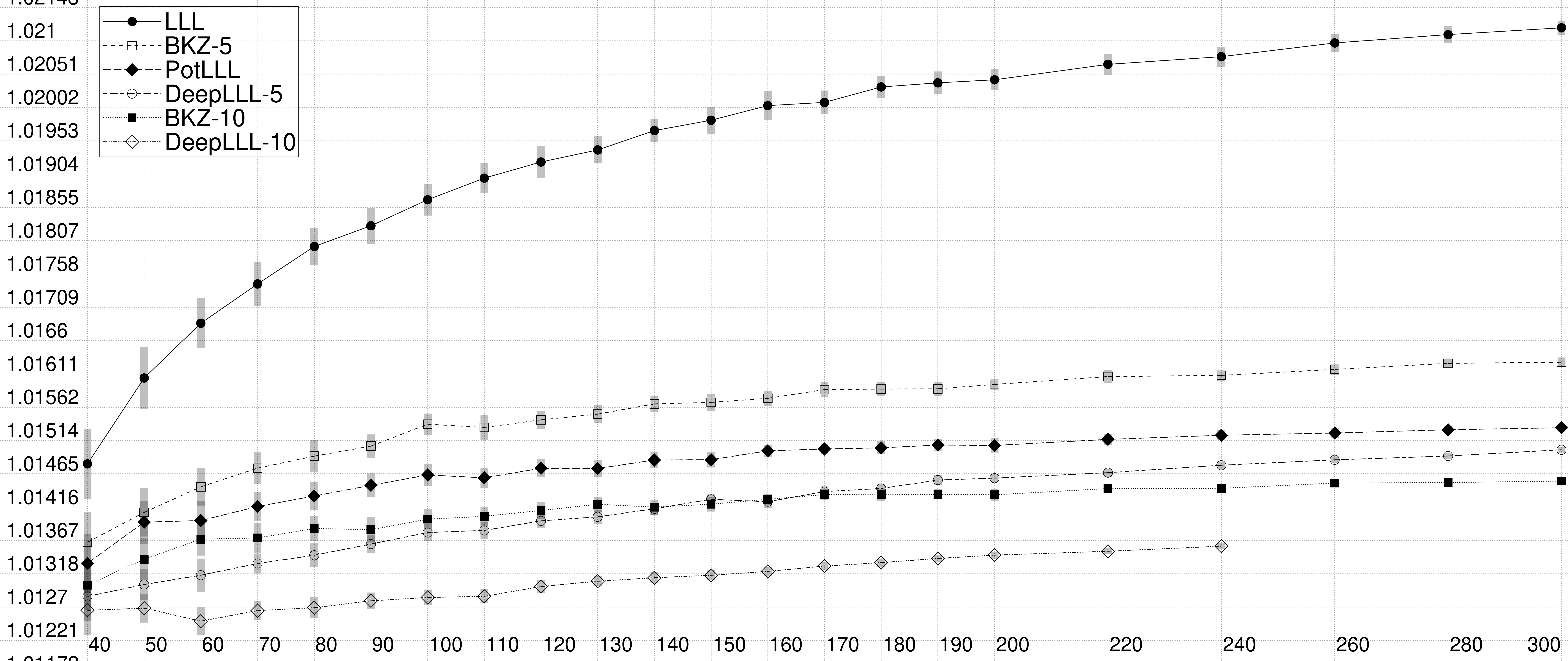}
      \end{center}
      \caption{Average $n$-th root Hermite factor.\vspace{0.3cm}}
      \label{fig:overview:apfa}
    \end{subfigure}
    \begin{subfigure}[t]{\textwidth}
      \begin{center}
        \includegraphics[width=\textwidth]{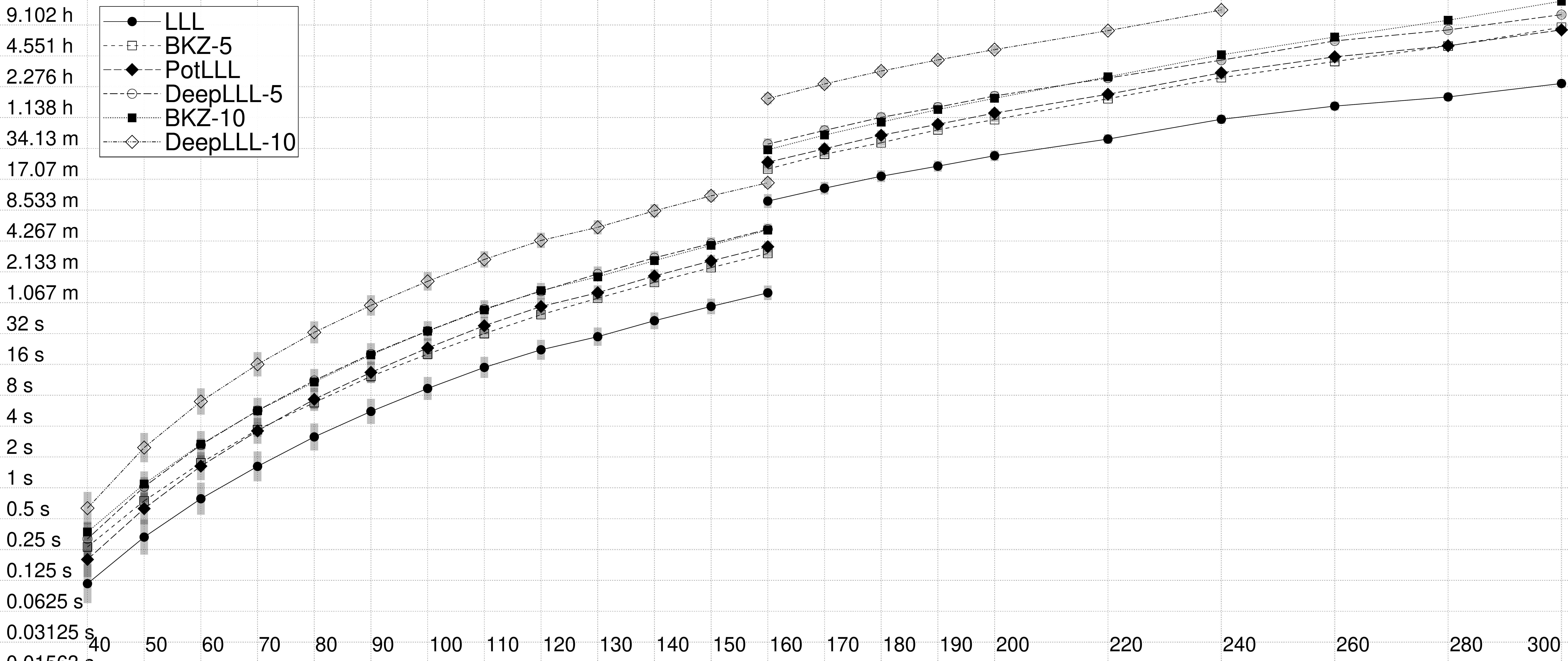}
      \end{center}
      \caption{Average logarithmic CPU time.}
      \label{fig:overview:time}
    \end{subfigure}
    \caption{Overview of performance of the algorithms for dimensions~$n$ ($x$ axis) from 40 to 300
      (using MPFR for $n \ge 160$).}
    \label{fig:overview}
  \end{center}
\end{figure}

\begin{figure}[h]
  \begin{center}
    \begin{subfigure}[t]{\textwidth}
      \begin{center}
        \includegraphics[width=\textwidth]{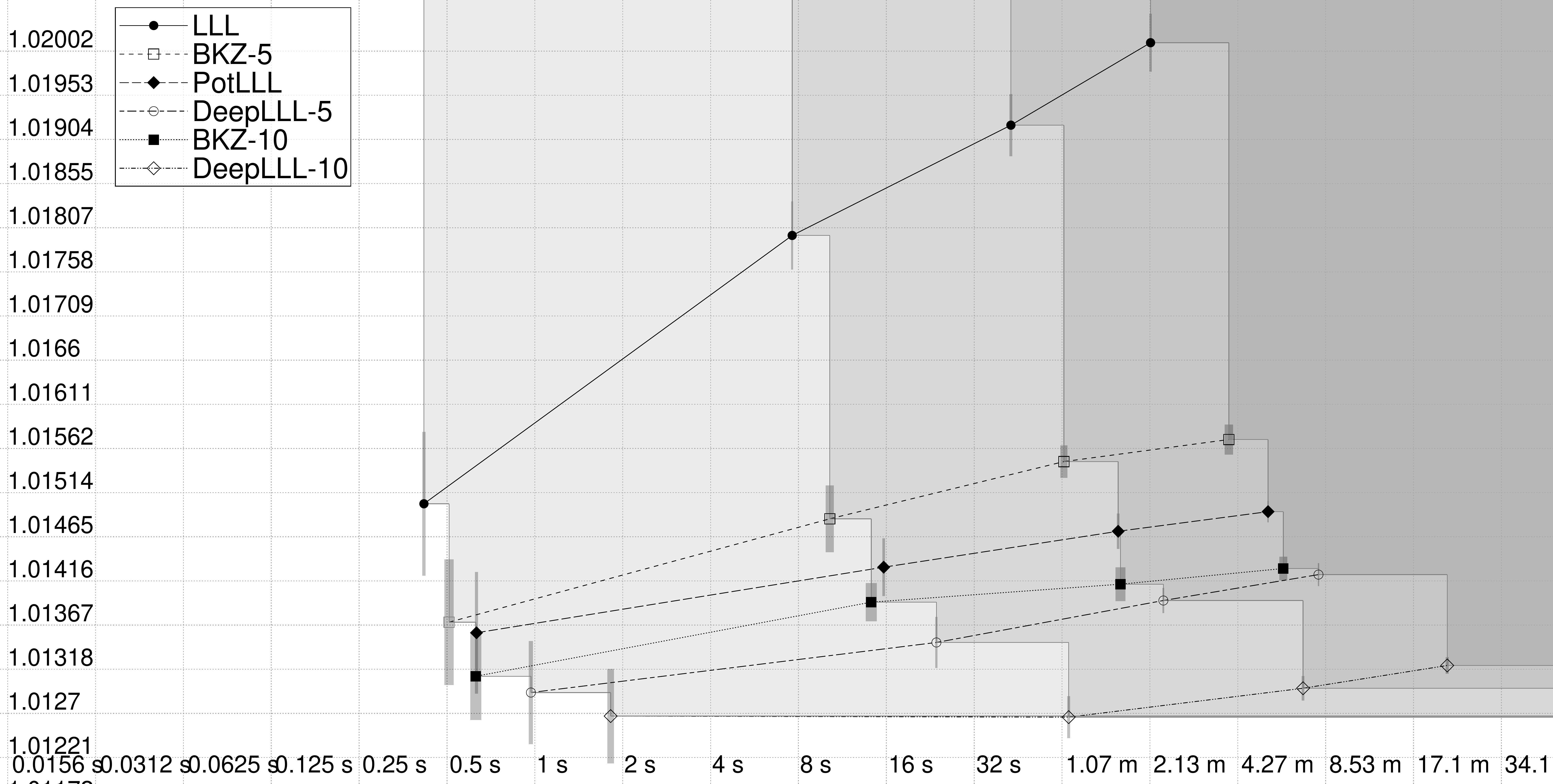}
      \end{center}
      \caption{\texttt{long double} arithmetic. The highlighted areas represent dimensions 40, 80,
      120 and 160.\vspace{0.3cm}}
      \label{fig:comparism:up:ld}
    \end{subfigure}
    \begin{subfigure}[t]{\textwidth}
      \begin{center}
        \includegraphics[width=\textwidth]{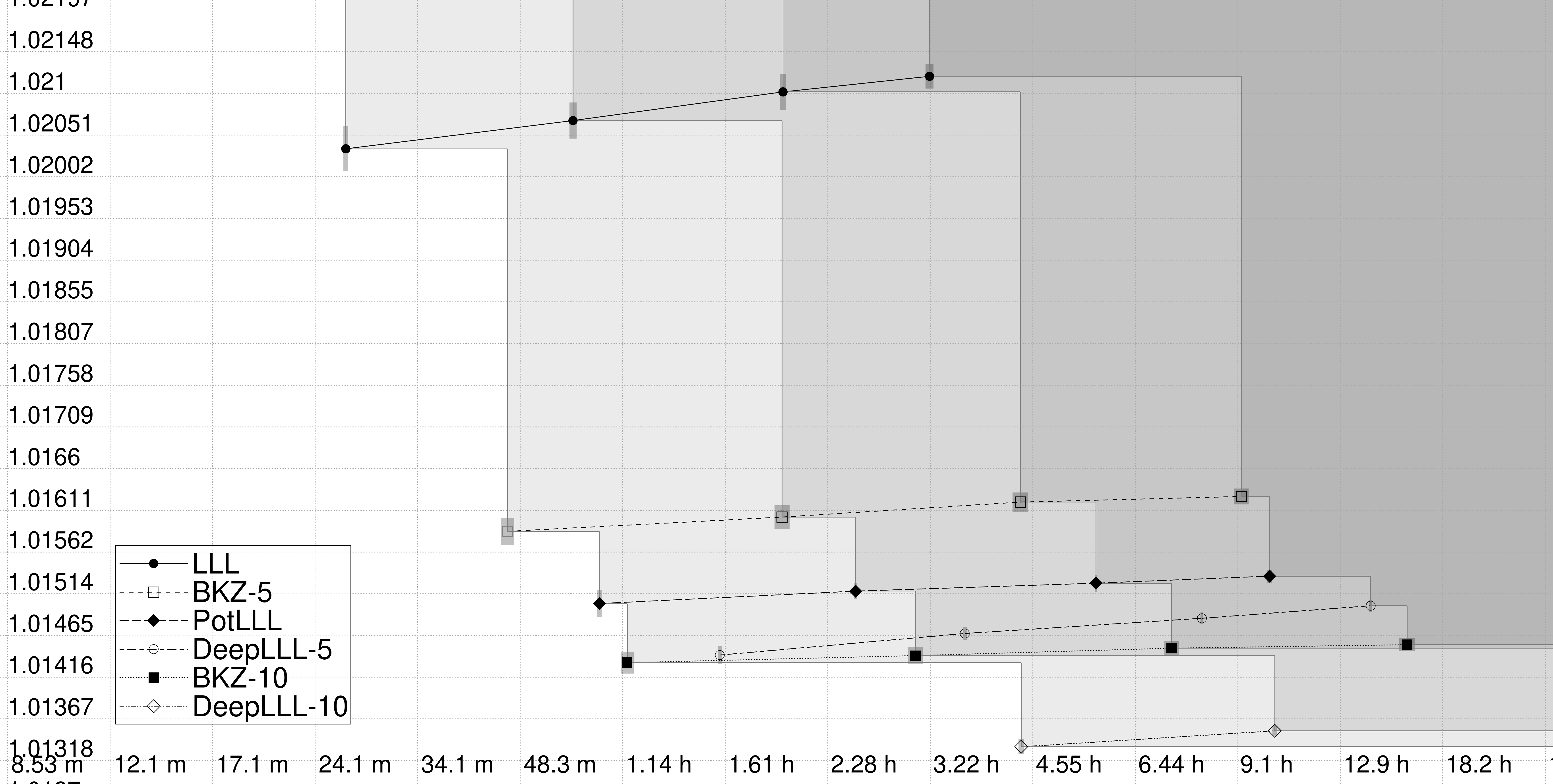}
      \end{center}
      \caption{MPFR arithmetic. The highlighted areas represent dimensions 180, 220, 260 and 300.}
      \label{fig:comparism:up:real}
    \end{subfigure}
    \caption{Comparison of $n$-th root Hermite factor ($y$ axis) vs. running times ($x$ axis) for
    both arithmetics for the original bases.}
    \label{fig:comparism:up}
  \end{center}
\end{figure}

\begin{figure}[h]
  \begin{center}
    \begin{subfigure}[t]{\textwidth}
      \begin{center}
        \includegraphics[width=\textwidth]{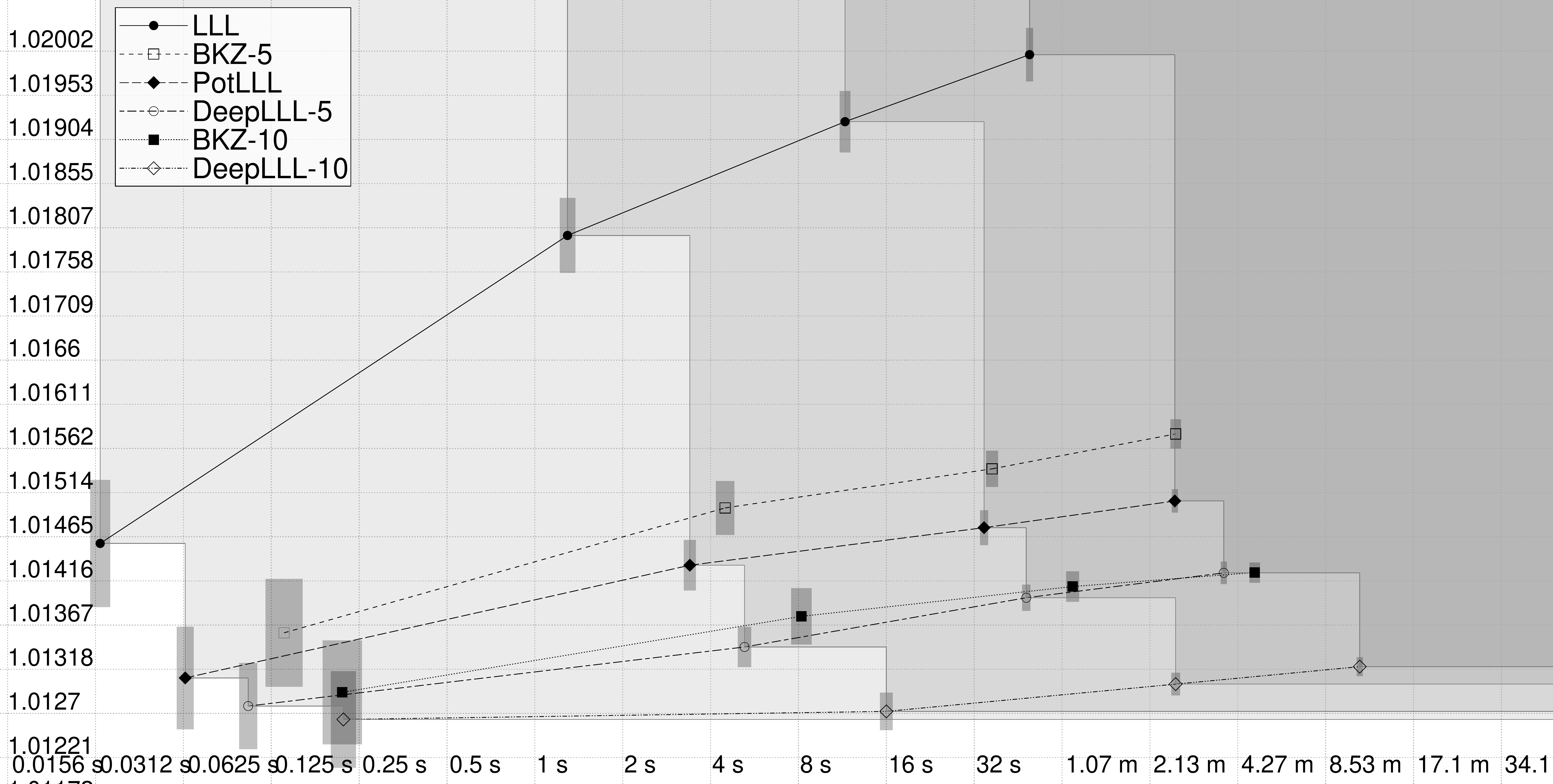}
      \end{center}
      \caption{\texttt{long double} arithmetic. The highlighted areas represent dimensions 40, 80,
      120 and 160.\vspace{0.3cm}}
      \label{fig:comparism:pp:ld}
    \end{subfigure}
    \begin{subfigure}[t]{\textwidth}
      \begin{center}
        \includegraphics[width=\textwidth]{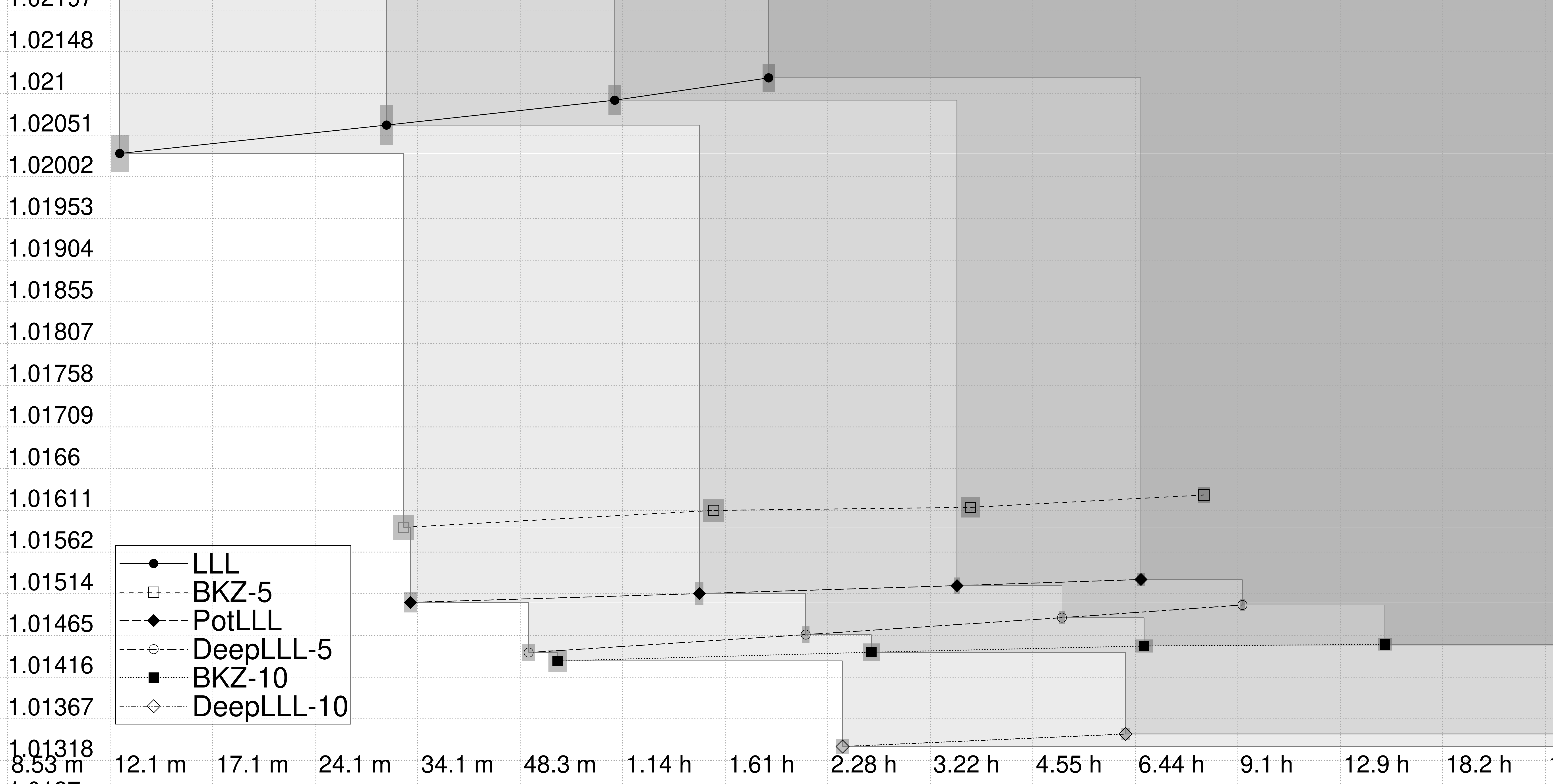}
      \end{center}
      \caption{MPFR arithmetic. The highlighted areas represent dimensions 180, 220, 260 and 300.}
      \label{fig:comparism:pp:real}
    \end{subfigure}
    \caption{Comparison of $n$-th root Hermite factor ($y$ axis) vs. running times ($x$ axis) for
    both arithmetics for preprocessed bases (0.75-LLL-reduced bases).}
    \label{fig:comparism:pp}
  \end{center}
\end{figure}
\section{Conclusion}\label{sec:conclusion}

We present a first provable polynomial time variant of Schnorr and Euchner's \dlll. While the
provable bounds are not better than for classical LLL -- in fact, for reduction parameter~$\delta =
1$, the existence of critical bases shows that better bounds do not exist -- the practical behavior
is much better than for classical LLL. We see that the $n$-th root Hermite factor of an
$n$-dimensional basis output by \plll in average does not exceed $1.0153^n$ for $n \le 300$.

For unprocessed random bases in Hermite normal form, \plll even outperforms
BKZ-5. Our experiments also
show that for such bases, \dlll with $\beta = 5$ outperforms BKZ-10. On the other hand, for bases
which are already reasonably preprocessed, for example by applying 0.75-LLL to a basis in Hermite
normal form, our algorithm is only slightly faster and sometimes even slower than BKZ-10, while
producing longer vectors.

It is likely that the improvements of the $L^2$ algorithm~\cite{ng06} and the $\tilde{L^1}$
algorithm~\cite{NovocinSV11} can be used to improve the runtime of our \plll\ algorithm, in order to
achieve faster runtime. We leave this for future work.

Moreover, deep insertions can be used together with BKZ as well. In particular, potential minimizing
deep insertions can be used. We added classical deep insertions and potential minimizing deep
insertions to BKZ. First experiments up to dimension~120 suggest that with regard to the output
quality, BKZ-5 with potential minimizing deep insertions is better than \plll, but worse than BKZ-5
with classical deep insertions, which in turn comes close to BKZ-10. BKZ-10 with potential
minimizing deep insertions is close to \dlll with $\beta = 5$, and BKZ-10 with classical deep
insertions close to \dlll with $\beta = 10$. For dimensions around 100, the speed of similarly
performing algorithms also behaves similarly.

\paragraph*{Acknowledgements}
This work was supported by CASED (\url{http://www.cased.de}). Michael Schneider is
supported by project BU~630/23-1 of the German Research Foundation (DFG). Urs Wagner and Felix
Fontein are supported by SNF grant no.~132256.

\end{document}